\documentclass[12pt,a4paper]{article}

\usepackage{amsmath, amsfonts}
\usepackage{amssymb, graphicx}
\usepackage{amsthm}
\usepackage[english]{babel}
\usepackage{color}
\usepackage{hyperref}
\usepackage{epsfig}
\usepackage{amsmath}
\usepackage{natbib}
\usepackage{enumerate}
\usepackage{setspace}
\usepackage{doi}
\hypersetup{colorlinks,
            linkcolor=blue,
            anchorcolor=blue,
            citecolor=blue,urlcolor={blue}}

\newtheorem{theorem}{Theorem}

\newtheorem{lemma}[theorem]{Lemma}

\newtheorem{proposition}[theorem]{Proposition}
\theoremstyle{definition}

\newtheorem{algorithm}[theorem]{Algorithm}
\newtheorem{example}[theorem]{Example}

\usepackage{pgf,tikz}
\usetikzlibrary{arrows}

\title{Open shop scheduling games\thanks{A. Atay acknowledges support from the Hungarian National Research, Development and Innovation Office via the grant PD-128348, and the Hungarian Academy of Sciences via the Cooperation of Excellences Grant (KEP-6/2019). P. Calleja acknowledges the support from research grant ECO2017-86481-P (Agencia Estatal de Investigaci\'{o}n (AEI) y Fondo Europeo de Desarrollo Regional (FEDER)) and  2017SGR778 (Generalitat de Catalunya). This work has been partly supported by COST Action CA16228 European Network for Game Theory.}}
\author{Ata Atay\thanks{Corresponding author. Institute of Economics, Hungarian Academy of Sciences. E-mail: \href{mailto:ata.atay@mta.hu}{ata.atay@krtk.mta.hu}} \and Pedro Calleja\thanks{Department of Economic, Financial, and Actuarial Mathematics, and BEAT, University of Barcelona, Spain.E-mail: \href{mailto:calleja@ub.edu}{calleja@ub.edu}} \and Sergio Soteras\thanks{Open University of Catalonia. E-mail: \href{mailto:ssoterass@uoc.com}{ssoterass@uoc.com}}}

\date{\today}

\begin{document}

\maketitle

\begin{abstract}
This paper takes a game theoretical approach to open shop scheduling problems with unit execution times to minimize the sum of completion times. By supposing an initial schedule and associating each job (consisting in a number of operations) to a different player, we can construct a cooperative TU-game associated with any open shop scheduling problem. We assign to each coalition the maximal cost savings it can obtain through admissible rearrangements of jobs' operations. By providing a core allocation, we show that the associated games are balanced. Finally, we relax the definition of admissible rearrangements for a coalition to study to what extend balancedness still holds.  \vspace{0.2cm} \vspace{0.2cm}\vspace{0.2cm}
\end{abstract}

\section{Introduction}
In a scheduling problem a set of jobs have to be executed by a number of machines. Such a general formulation arises in many real life situations, like manufacturing processes, computer science, logistics, etc. In this paper we consider \textit{open shop scheduling problems} introduced by \cite{gs76} in which $n$ jobs consisting of $m$ operations have to be processed on $m$ machine, each operation on a different machine. We do not allow preemptions, the order in which jobs' operations are processed is immaterial but two operations of the same job cannot be processed simultaneously (for a survey see Chapter 8 in \citealp{p12} or Chapter 6 in \citealp{l04}).

By assuming that every job belongs to a player, that incurs some waiting cost until she can leave the system and that there is an initial processing schedule (let say first come, first served) we can take a game theoretical approach. The main question is how to distribute the cost savings the players can obtain by cooperation, whenever they rearrange their jobs to be processed in an optimal way, minimizing total waiting costs. 

\cite{cetal89} are the first to study one-machine situations with weighted completion time as the cost-criterion from such cooperative point of view. In order to obtain stable allocations of the total cost savings, where no coalition receives less than the cost savings they can generate by themselves, we need to determine first what rearrangements of their jobs' operations are allowed for the coalition. An accepted and broadly used definition of admissible rearrangement for a coalition (\citealp{cetal89}) imposes that the set of predecessors of a player not in the coalition on a machine should be the same as initially. In our model, we also consider the weighted completion time as the cost-criterion, but contrary to the most of the literature in the field and due to machines may incur idle time, such condition does not prevent players not in the coalition of being hurt. This forces us to impose additional conditions on what should be admissible for a coalition.

\cite{cetal02} provide an extensive review of many scheduling problems that has been treated from this point of view. In particular, many different multiple machine problems has been studied as parallel machines (\citealp{hetal99}; \citealp{caetal02}) or flow shop problems (\citealp{vdnetal92}; \citealp{eetal08}).

The computational complexity of finding optimal schedules for open shop problems to minimize the weighted sum of completion times has been well established in the literature from \cite{ac82}. However, \cite{aa84} provide two clear algorithms to obtain optimal schedules for unit open shop scheduling problems, where all processing times of all operations are equal and all players have the same linear cost function. In the main result of the paper we provide a stable allocation of the total cost savings obtained by cooperation for unit open shop scheduling problems. Finally, we study to what extend such allocation is still stable if we relax the definition of admissible rearrangements for a coalition. We follow the same approach of \cite{cetal93} and \cite{s06} where operations of jobs in a coalition are allowed to jump over operations of jobs not belonging to the coalition  whenever this does not hurt the interest (completion time) of jobs outside the coalition. We obtain some positive results depending on the specifications of such relaxations. On the other hand, stable allocations may not exist for weaker relaxation conditions. \cite{vh03} and \cite{metal15} consider different relaxation approaches.

The rest of the paper is organized as follows. In Section \ref{sec:open shop problems} we introduce unit time open shop problems and present some optimal schedules for the weighted completion time criterion. Section \ref{sec:model} introduces the coalitional game associated with an open shop scheduling problem with initial schedule and discusses which rearrangements should be admissible for a coalition. In Section \ref{sec:core}, we show our main result, that is, the core of a unit open shop scheduling game is non-empty. Finally, in Section \ref{sec:relax}, we study to what extend balancedness still holds when we relax the definition of admissible rearrangements for a coalition.
 
\section{Open shop scheduling problems}
\label{sec:open shop problems}
An \textit{open shop scheduling problem} consists of $n$ jobs, $N=\{1,\ldots , n\}$, each of them consisting of $m$ operations each one to be processed on a different machine, being $M=\{1,2,\ldots,m\}$ the set of machines. When no confusion arises we denote by $|N|=n$ the cardinality of the set of jobs and by $|M|=m$ the cardinality of the set of machines. Alternatively, we can think of a set of players, any of them needing to finish a job that consists of $m$ operations, each of these operations to be processed on a different machine. Players and jobs are used interchangeably throughout this paper. 

The operation of job $i\in N$ on machine $j\in M$ is denoted by the pair $(i,j)$ and $p^{j}_{i}>0$ denotes the processing time of $(i,j)$. We assume that all operations have to be processed uninterrupted, that is, preemptions are not allowed. Moreover, in an open shop scheduling problem the process order of a job's operations is immaterial, but two operations of the same job cannot be processed simultaneously. Also, a machine cannot process more than one job at a time. 

A \emph{schedule} is a mapping $s: N\times M\rightarrow \mathbb{R}_{+}$ that assigns to every operation a starting time. The set of all feasible schedules, according to the open shop specifications, is denoted by $\mathcal{S}$. Let $s\in\mathcal{S}$, we denote the starting time of operation $(i,j)$ according to the schedule $s$ by $t^{j}_{i}(s):=s(i,j)$. Since no preemption is allowed, $C^{j}_{i}(s)=t^{j}_{i}(s)+p^{j}_{i}$ is the completion time of the operation $(i,j)$ according to $s$. We denote the completion time of job $i$ according to $s$ by $C_{i}(s)=\max\limits_{j\in M} C^{j}_{i}(s)$. 

A \emph{scheme} $\sigma=(\sigma^{j})_{j\in M}$ is a collection of $m$ bijections, $\sigma^{j}:N\rightarrow\{1,\ldots , n\}$, one for each machine $j\in M$, where $\sigma^{j}(i)=k$ interprets the operation of job (player) $i$ on machine $j$ is at position $k$ according to scheme $\sigma$. In other words, player $i$ has the right to process her operation on machine $j$ before her $n-k$ followers according to $\sigma^{j}$. The set of all possible schemes is denoted by $\Sigma$. A feasible schedule $s\in\mathcal{S}$ is \textit{compatible} with the scheme $\sigma\in\Sigma$ if and only if for all $j\in M$ and $i,i'\in N$ it holds $$t^{j}_{i}(s)<t^{j}_{i'}(s)\iff \sigma^{j}(i)<\sigma^{j}(i').$$ 

In the next example, we illustrate that a given scheme $\sigma\in\Sigma$ admits a number of different compatible admissible schedules. On the other hand, a given schedule $s\in\mathcal{S}$ is clearly compatible with a unique scheme.

\begin{example}
\label{ex:easy}
Consider the open shop scheduling problem with $N=\{1,2\}$, $M=\{1,2\}$, $p^{j}_{i}=1$ for all $i\in N$ and for all $j\in M$, and consider the scheme $\sigma^{1}=\sigma^{2}=(1,2)$. Then, $\sigma$ admits, among others, the following two feasible schedules $s_{1}$ and $s_{2}$:
\begin{center}
\begin{tabular}{ c|c|c|c|c } 
\cline{2-4} 
$m_{1}$  & 1 & 2 &   & $s_1$,\\  \cline{2-4} 
$m_{2}$ &   & 1  & 2  &\\  \cline{2-4}
\end{tabular}
\end{center}
\begin{center}
\begin{tabular}{ c|c|c|c|c } 
\cline{2-4} 
$m_{1}$  & & 1 & 2  & $s_2$.\\  \cline{2-4} 
$m_{2}$ &  1 & 2 &   &\\  \cline{2-4}
\end{tabular}
\end{center}
In the first schedule machine 2 incurs idle time, while in the second schedule machine 1 incurs idle time.
\end{example}

A \textit{semi-active schedule} is such that there does not exist an operation which could be started earlier without altering the processing scheme or violating the restrictions on the processing of operations, according to the open shop specifications. So, all machines start processing all operations as soon as it is possible without violating the fact that two operations of the same job cannot be processed at the same time.  In Example~\ref{ex:easy}, only the feasible schedules $s_{1}$ and $s_{2}$ are semi-active. Observe that there is no one-to-one correspondence between schemes and semi-active schedules for open shop problems.

Every job (player) $i\in N$ has a waiting cost that is linear with respect to the moment it can leave the system, i.e. the cost function of a job $i\in N$ for a given $s\in\mathcal{S}$ is of type $c_{i}(s)=\alpha_{i}C_{i}(s)$ where $\alpha_{i}>0$ is the weight or waiting cost per unit time of player $i$. Our first aim is to find an optimal schedule $\hat{s}_{N}\in\mathcal{S}$ that minimizes the weighted sum of completion times. Note that since the waiting costs are non-decreasing with respect to the completion time for all $i\in N$, then we only need to look at semi-active schedules.

Finding optimal schedules for such open shop situations is a difficult computational problem and has been proved to be NP-hard even if there are only two machines (see for instance \citealp{ac82}). Thus, henceforth we restrict our attention to unit  time open shop problems. In a unit time open shop problem $p^{j}_{i}=p$ for all $i\in N$ and all $j\in M$, and $\alpha_{i}=\alpha$ for all $i\in N$. Without loss of generality we assume that $p^{j}_{i}=1$ for all $i\in N$ and all $j\in M$, and $\alpha_{i}=1$ for all $i\in N$. A unit open shop scheduling problem is a pair $(N,M)$.

\cite{aa84} provide two different optimal schedules for unit time open shop problems minimizing the weighted sum of completion times. For our purpose, here we introduce one of them:

\begin{algorithm}[\citealp{aa84}]
\label{Algo:AA}
Schedule operations of player $i\in N$ continuously, starting at the earliest possible time (respecting operations processing restrictions) on machine $i\mod (m)$\footnote{For all $x,y\in \mathbb{R}, \lfloor x \rfloor :=\max\{k\in\mathbb{Z}\mid k\leq x\}$,  $\lceil x \rceil :=\min\{ k\in\mathbb{Z} \mid x\leq k\}$, and  $x \text{ mod } (y):= x-y\lfloor \frac{x}{y}\rfloor$.} until machine $m$. Then, move to the earliest possible time (respecting operations processing restrictions) on machine 1 and schedule continuously the remaining operations of player $i$.
\end{algorithm} 

Next, we provide a unit time open shop scheduling problem with six players and four machines to illustrate Algorithm \ref{Algo:AA}.

\begin{example}
\label{ex:algo}
Consider $(N,M)$ with $N=\{1,2,3,4,5,6\}$ and $M=\{1,2,3,4\}$. Then, an optimal schedule $\hat{s}_{N}$ according to Algorithm \ref{Algo:AA} is: 

\begin{center}
\begin{tabular}{ c|c|c|c|c|c|c|c|c| } 
 \cline{2-9} 
$m_{1}$  & 1 & 4 & 3 & 2 & 5 & & & 6 \\ \cline{2-9} 
$m_{2}$  & 2 & 1 & 4 & 3 & 6 & 5 & & \\ \cline{2-9} 
$m_{3}$  & 3 & 2 & 1 & 4 & & 6 & 5 & \\ \cline{2-9} 
$m_{4}$  & 4 & 3 & 2 & 1 & & & 6 & 5 \\ 

\cline{2-9} 
\end{tabular},
\end{center}
which is only compatible with the associated scheme $\sigma$:
\begin{center}
\begin{tabular}{ c|c|c|c|c|c|c| } 
\cline{2-7} 
$\sigma^{1}$ & 1 & 4 & 3 & 2 & 5 & 6 \\  \cline{2-7} 
$\sigma^{2}$ & 2 & 1 & 4 & 3 & 6 & 5 \\  \cline{2-7} 
$\sigma^{3}$ & 3 & 2 & 1 & 4 & 6 & 5 \\  \cline{2-7} 
$\sigma^{4}$ & 4 & 3 & 2 & 1 & 6 & 5 \\  \cline{2-7} 
\end{tabular}.
\end{center}
As noted in \cite{aa84}  if $n=mk+l$ with $k=\lfloor \frac{n}{m}\rfloor$ and $l\geq 0$, this algorithm constructs $k$ compact blocks where machines do not stop between operations. In block $1\leq r\leq k$, $m$ jobs start processing at time $(r-1)m$ and finishes at time $rm$. In the last block $k+1$, the last $l$ jobs start at $km$ and finishes at $(k+1)m$, but in this block the machines incur some idle interval. 

Note that there is a machine $(m_{2})$ that processes all operations continuously, and $C_{i}(\hat{s}_{N})=\left\lceil \frac{\sigma^{2}(i)}{m}\right\rceil m$ for all $i\in N$, indeed. Note also that in fact there are many optimal schedules which can be obtained by just switching the names of the players. 
\end{example}

\section{Unit time open shop scheduling games}
\label{sec:model}
Under the assumption that there is an initial feasible schedule $s_{0}\in\mathcal{S}$ that describes the initial processing of the operations on all machines, a \emph{unit time open shop scheduling problem} with initial schedule $s_0$ is a triplet $(N,M,s_{0})$.

A \emph{cooperative transferable utility (TU) game} is defined by a pair $(N,v)$ where $N$ is the (finite) player set and the characteristic function $v$ assigns a real number $v(T)$ to each coalition $T\subseteq N$, with $v(\emptyset)=0$. 

For any coalition $\emptyset\neq T\subseteq N$ and any feasible schedule $s\in\mathcal{S}$, by $c_{T}(s)=\sum\limits_{i\in T}c_{i}(s)$ we denote the waiting cost of the coalition $T$ according to $s$. Then, given a unit open shop scheduling problem with initial schedule $(N,M,s_{0})$, we define the unit time open shop scheduling game $(N,v)$ where the characteristic function assigns to every coalition the maximal cost savings it can obtain by means of admissible rearrangements (or admissible schedules). That is, if $\mathcal{AS}(T)\subseteq\mathcal{S}$ denotes the set of admissible schedules for coalition $T\subseteq N$,

$$v(T)=c_{T}(s_{0})-c_{T}(\hat{s}_{T}),$$
where $\hat{s}_{T}\in\mathcal{AS}(T)$ is such that $c_{T}(\hat{s}_{T})=\min\limits_{s\in\mathcal{AS}(T)}c_{T}(s)$.

Clearly, $\mathcal{AS}(N)$ should coincide with $\mathcal{S}$ under any definition of admissible rearrangement. \cite{cetal93} impose two principles that should be considered when defining which rearrangements are admissible for a coalition:
\begin{enumerate}[(i)]
\item The rearrangement should not hurt the interests of the players outside the coalition.
\item The rearrangement should be possible without an active cooperation of players outside the coalition.
\end{enumerate} 

Following most of the literature on one or multiple parallel machines (see for instance \citealp{cetal02}) we say that a schedule $s$, that is compatible with the unique scheme $\sigma$, will be admissible for a coalition $\emptyset\neq T\subset N$ if for each machine no player outside the coalition $T$ has a different set of predecessors as initially. That is, if by $\sigma_{0}$ we denote the unique scheme compatible with $s_{0}$, for all $i\in N\setminus T$ and all $j\in M$ it holds
\begin{equation}
\label{relax1}
\{k\in N : \sigma^{j}(k)<\sigma^{j}(i)\}=\{k\in N : \sigma^{j}_{0}(k)<\sigma^{j}_{0}(i)\}.
\end{equation}

Hence, for a given $j\in M$,  switches are only allowed among players from connected coalitions. A coalition $T\subseteq N$ is called connected with respect to $\sigma_{0}^{j}$ if for all $i,i'\in T$ and $k$ such that $\sigma_{0}^{j}(i)<\sigma_{0}^{j}(k)<\sigma_{0}^{j}(i')$ it holds that $k\in T$. We denote by $T/\sigma^{j}_{0}$ the set of maximally connected components of $T$ according to $\sigma^{j}_{0}$.

We denote the set of admissible schedules for coalition $T$ that satisfies (\ref{relax1}) by $\mathcal{AS}^{1}(T)$. Unfortunately, as shown in Example \ref{ex:hurt}, given $T\subseteq N$, $\mathcal{AS}^{1}(T)$ might include admissible rearrangements that hurts players outside the coalition $T$.

\begin{example}
\label{ex:hurt}
Consider $(N,M,s_{0})$ with $N=\{1,2,3,4,5\}$, $M=\{1,2\}$, 
and the initial schedule $s_{0}$ as follows:
\begin{center}
\begin{tabular}{ c|c|c|c|c|c|c| } 
\cline{2-7} 
$m_{1}$ & 1 & 2 &  & 3 & 4 & 5 \\  \cline{2-7} 
$m_{2}$ & 5 & 1 & 3 & 4 & 2 &  \\  \cline{2-7}
\end{tabular}.
\end{center}
Let $T=\{3,5\}$. It is easy to check that $\hat{s}_{T}\in\mathcal{AS}^{1}(T)$ is: 
\begin{center}
\begin{tabular}{ c|c|c|c|c|c|c| } 
\cline{2-7} 
$m_{1}$ & 1 & 2 & 3 & 4 & 5 & \\  \cline{2-7} 
$m_{2}$ & 5 & 1 & & 3 & 4 & 2   \\  \cline{2-7}
\end{tabular}.
\end{center}
Then, $C_{3}(s_{0})=4$, $C_{5}(s_{0})=6$ while $C_{3}(\hat{s}_{T})=4$, $C_{5}(\hat{s}_{T})=5$, and $v(\{3,5\})=1$. However, one can easily see that $\hat{s}_{T}$ hurts the interests of player 2 who does not take part of the coalition $T$, since $C_{2}(s_{0})=5<6=C_{2}(\hat{s}_{T})$.
\end{example}

The following example shows that, moreover, given $T\subseteq N$, $\mathcal{AS}^{1}(T)$ might include an admissible rearrangement that requires the active cooperation of players outside the coalition $T$. 

\begin{example}
\label{ex:act_coop}
Consider $(N,M,s_{0})$ with $N=\{1,2,3\}$, $M=\{1,2\}$, and the initial schedule $s_{0}$ as follows:
\begin{center}
\begin{tabular}{ c|c|c|c|c|c| } 
\cline{2-6} 
$m_{1}$ & 1 & 2 & 3 & & \\  \cline{2-6} 
$m_{2}$ & & 1 & & 3 & 2 \\  \cline{2-6}
\end{tabular}.
\end{center}
Let $T=\{2\}$. It is easy to check that $\hat{s}_{\{2\}}\in\mathcal{AS}^{1}(T)$ is:
\begin{center}
\begin{tabular}{ c|c|c|c|c|c| } 
\cline{2-5} 
$m_{1}$ & 1 & 2 & & 3 \\  \cline{2-5} 
$m_{2}$ & & 1 & 3 & 2 \\  \cline{2-5}
\end{tabular}.
\end{center}
Then, $C_{2}(s_{0})=5$, $C_{2}(\hat{s}_{\{2\}})=4$, and hence $v(\{2\})=1$. Even though $\hat{s}_{\{2\}}$ delays the operation $(3,1)$, $C_{3}(s_{0})=C_{3}(\hat{s}_{\{2\}})$, and hence player 3 is not hurt. However, changing from schedule $s_{0}$ to $\hat{s}_{\{2\}}$ requires the active cooperation of player 3. Observe that contrary to $s_{0}$, in $\hat{s}_{\{2\}}$ player 3 decides to process first $(3,2)$ instead of $(3,1)$.
\end{example}

In view of Examples \ref{ex:hurt} and \ref{ex:act_coop}, it is clear that due to the fact that coalitions can make use of idle times on machines, condition (\ref{relax1}) is not enough to guarantee the two principles required for the definition of admissible rearrangements of a coalition.

In \cite{cetal93} a number of different approaches to admissible arrangements are studied. They combine two different ideas. In the first one, players in a coalition are allowed to jump over players outside the coalition. We address this approach in Section \ref{sec:relax}. In the second one, they simply focus on the starting time (completion time) of operations of players outside the coalition. If those times do not increase, they will not be worse off. We present three proposals, inspired by those in \cite{cetal93}. The first one is based on the starting time of operations. That is, a schedule $s\in\mathcal{S}$, with corresponding compatible scheme $\sigma$, is admissible for coalition $\emptyset\neq T\subset N$ if it satisfies (\ref{relax1}) and the starting time of operations of players outside $T$ remains unchanged:
\begin{equation}
\label{relax2}
\text{for all}\,\, i\in N\setminus T \text{ and all }\, j\in M, \, \text{it holds that}, t^{j}_{i}(s)=t^{j}_{i}(s_{0}).
\end{equation}
We denote the set of admissible schedules for coalition $T$ that satisfies (\ref{relax1}) and (\ref{relax2}) by $\mathcal{AS}^{2}(T)$.

A second approach considers that a schedule $s\in\mathcal{S}$ with corresponding compatible scheme $\sigma$ is admissible for a coalition $\emptyset\neq T\subset N$ if it satisfies (\ref{relax1}) and the starting time of operations of players outside $T$ does not increase:
\begin{equation}
\label{relax3}
\text{for all}\,\, i\in N\setminus T \text{ and all }\, j\in M, \, \text{it holds that}, t^{j}_{i}(s)\leq t^{j}_{i}(s_{0}).
\end{equation}
We denote the set of admissible schedules for coalition $T$ that satisfies (\ref{relax1}) and (\ref{relax3}) by $\mathcal{AS}^{3}(T)$.

In \cite{cetal93} only one machine problems are studied, hence condition (\ref{relax3}) is equivalent to enforcing non-increasing completion times for all $i\in N\setminus T$. Following that spirit, we introduce a new definition of admissible rearrangements. A schedule $s\in\mathcal{S}$ with corresponding compatible scheme $\sigma$ is admissible for a coalition $\emptyset\neq T\subset N$ if it satisfies (\ref{relax1}) and the completion time of players outside $T$ does not increase:
\begin{equation}
\label{relax4}
\text{for all}\,\, i\in N\setminus T, \, \, \text{it holds that}, C_{i}(s)\leq C_{i}(s_{0}).
\end{equation}
We denote the set of admissible schedules for coalition $T$ that satisfies (\ref{relax1}) and (\ref{relax4}) by $\mathcal{AS}^{4}(T)$.

Clearly, for a given $\emptyset\neq T\subset N$, we have $\mathcal{AS}^{2}(T)\subseteq \mathcal{AS}^{3}(T)\subseteq\mathcal{AS}^{4}(T)$. Moreover, conditions (\ref{relax2}), (\ref{relax3}), and (\ref{relax4}) ensure that admissible rearrangements will not hurt the interests of players outside the coalition. 

On the other hand, one can easily check that in Example \ref{ex:act_coop}, $\hat{s}_{\{2\}}\in\mathcal{AS}^{4}(\{2\})$ and, as noted, changing from $s_{0}$ to $\hat{s}_{\{2\}}$ requires the active cooperation of player 3. Observe that this is possible because player 3 ``makes use'' of the idle times on machines. Given $s_{0}\in\mathcal{S}$ and $T\subset N$, a player $i\in N\setminus T$ can only ``make use'' of idle times in profit of coalition $T$ to reach a rearrangement $s\in\mathcal{AS}^{4}(T)$, as player 3 in Example \ref{ex:act_coop}, if there is a machine $j\in M$ such that the initial starting time of operation $(i,j)$, $t^{j}_{i}(s_{0})$, is smaller than the starting time of $(i,j)$ according to $s$; $$t^{j}_{i}(s_{0})<t^{j}_{i}(s).$$ In Example \ref{ex:act_coop}, such machine is $m_1$. Hence, obviously $s\notin \mathcal{AS}^{3}(T)$. Then, admissible rearrangements in $\mathcal{AS}^{2}(T)$ and $\mathcal{AS}^{3}(T)$ do not allow for the active cooperation of players outside $T$.

For every different set of admissible rearrangements, we can associate a cooperative TU-game. Let $(N,v^{k})$ denote the game where the set of admissible rearrangements for a coalition $T\subseteq N$ is $\mathcal{AS}^{k}(T)$, with $k=\{2,3,4\}$. Next, we provide a relationship between the cooperative games defined.
\begin{proposition}
\label{pro:relation_games}
Let $(N,M,s_0)$ be a unit open shop scheduling problem with initial schedule. Then, it holds
\begin{align*}
v^{2}(N)&=v^{3}(N)=v^{4}(N), & \text{and}  \\
v^{2}(T)&\leq v^{3}(T)\leq v^{4}(T) & \text{for all } T\subset N.
\end{align*}
\end{proposition}
\begin{proof}
It follows from the observation that $\mathcal{AS}^{2}(T)\subseteq\mathcal{AS}^{3}(T)\subseteq\mathcal{AS}^{4}(T)$ for all $T\subset N$ and $\mathcal{AS}^{2}(N)=\mathcal{AS}^{3}(N)=\mathcal{AS}^{4}(N)$.
\end{proof}
\section{Non-emptiness of the core}
\label{sec:core}
Given a cooperative game $(N,v)$, a \emph{payoff vector} $x\in\mathbb{R}^{N}$ represents the payoffs to the players. Each component $x_{i}$ is interpreted as the allotment to player $i\in N$. The total payoff to a coalition $S\subseteq N$ is denoted by $x(S)=\sum\limits_{i\in S}x_{i}$ with $x(\emptyset)=0$. In order to study the set of stable allocations of the total cost savings $N$ can obtain, we introduce the \emph{core} of a cooperative game $(N,v)$ that consists of those payoff vectors that satisfy efficiency and every coalition $S\subset N$ receives at least its worth: $x(S)\ge v(S)$ (\citealp{g59}). Formally, the core of a cooperative game $(N,v)$ is:
$$C(v)=\{x\in\mathbb{R}^{N}\mid x(N)=v(N),\quad x(S)\geq v(S)\quad \text{for all}\quad S\subset N\}.$$

A game is \emph{balanced} if it has a non-empty core. Given a unit open shop scheduling problem with initial schedule $(N, M, s_{0})$, it follows from Proposition \ref{pro:relation_games} that any core element of the game $(N,v^{4})$ is also a core element of the games $(N,v^{3})$ and $(N,v^{2})$. Hence, we will focus on the core of the cooperative game $(N,v^{4})$. If we show that the game $(N,v^{4})$ is balanced, then independently of the chosen definition of admissible rearrangements, the associated game will admit stable allocations of the total cost savings.

Convexity (\citealp{s71}) and $\sigma$-component additivity (\citealp{cetal94}) are conditions that have been extensively studied to prove balancedness of sequencing games associated with different sequencing problems (see for instance \citealp{cetal94}; \citealp{hetal95}; \citealp{betal02}; \citealp{metal18}). One of the requirements for a game $(N,v)$ to be $\sigma$-component additive is that $v(\{i\})=0$ for all $i\in N$. One can easily check that in Example \ref{ex:act_coop}, $\hat{s}_{\{2\}}\in\mathcal{AS}^{4}(\{2\})$ and consequently $v^{4}(\{2\})=1$. So, the game $(N,v^{4})$ is not $\sigma$-component additive. On the other hand, it is well-known that convexity implies superadditivity. A game $(N,v)$ is said to be \textit{superadditive} if $v(S\cup T)\geq v(S)+v(T)$ for all $S,T\subseteq N$, $S\cap T=\emptyset$. In the next 13-player example, we show that the game $(N,v^{4})$ arising from a unit time open shop scheduling problem with initial schedule $(N,M,s_0)$ need not be superadditive (nor convex). \footnote{Although Example \ref{ex:superadd} is a 13-player game, we can show that convexity does not hold for a 9-player game.}
\begin{example}
\label{ex:superadd}
Consider $(N,M,s_{0})$ with $N=\{1,2,3,4,5,6,7,8,9,10,11,12,13\}$, $M=\{1,2,3,4\}$, and the initial schedule $s_{0}$ as follows:
\begin{center}
\begin{tabular}{ c|c|c|c|c|c|c|c|c|c|c|c|c|c|c|c|c|c|c| }
 \cline{2-19} 
$m_{1}$  & 1 & 2 & 3 & 4 & 5 & 6 & 7 & 8 & 9 & 10 & 11 & 12 & 13 & & & & & \\ \cline{2-19} 
$m_{2}$  & 13 & 12 & 10 & 5 & 4 & & 3 & 1 & 2 & 8 & & 11 & 6 & 7 & 9 & & & \\ \cline{2-19} 
$m_{3}$  & 4 & 5 & 1 & 2 & 12 & 3 & 6 & 7 & 8 & 9 & & & 11 & 10 & 13 & & & \\ \cline{2-19} 
$m_{4}$  & 12 & 9 & 2 & 10 & 1 &  &  & 3 & 4 & 5 & & & & 11 & 6 & 7 & 8 & 13 \\ \cline{2-19} 
\end{tabular}.
\end{center}
Let $S=\{1,2\}$. It is easy to check that $\hat{s}_{\{1,2\}}\in\mathcal{AS}^{4}(\{1,2\})$ is:
\begin{center}
\begin{tabular}{ c|c|c|c|c|c|c|c|c|c|c|c|c|c|c|c|c|c| } 
 \cline{2-18} 
$m_{1}$  & 1 & 2 & 3 & 4 & 5 & 6 & 7 & 8 & 9 & 10 & 11 & 12 & 13 & & & & \\ \cline{2-18} 
$m_{2}$  & 13 & 12 & 10 & 5 & 4 & 3 & 1 & 2 & 8 & 11 & 6 & 7 & 9 & & & &  \\ \cline{2-18} 
$m_{3}$  & 4 & 5 & 1 & 2 & 12 & & 3 & 6 & 7 & 8 & 9 & 11 & 10 & 13 & & &  \\ \cline{2-18} 
$m_{4}$  & 12 & 9 & 2 & 10 & 1 &  & & 3 & 4 & 5 & & & 11 & 6 & 7 & 8 & 13 \\ \cline{2-18} 
\end{tabular}.
\end{center}
Hence, $v^{4}(\{1,2\})=2$. Now, let $T=\{4,5\}$. It is easy to check that $\hat{s}_{\{4,5\}}\in\mathcal{AS}^{4}(\{4,5\})$ is:
\begin{center}
\begin{tabular}{ c|c|c|c|c|c|c|c|c|c|c|c|c|c|c|c| } 
 \cline{2-16} 
$m_{1}$  & 1 & 2 & 3 & 4 & 5 & 6 & 7 & 8 & 9 & 10 & 11 & 12 & 13 & &  \\ \cline{2-16} 
$m_{2}$  & 13 & 12 & 10 & 5 & 4 & & 3 & 1 & 2 & 8  & & 11 & 6 & 7 & 9   \\ \cline{2-16} 
$m_{3}$  & 4 & 5 & 1 & 2 & 12 & & & 3 & 6 & 7 & 8 & 9 & 11 & 10 & 13  \\ \cline{2-16} 
$m_{4}$  & 12 & 9 & 2 & 10 & 1  & 3 & 4 & 5 & 11 & 6 & 7 & 8  & & 13 & \\ \cline{2-16} 
\end{tabular}.
\end{center}
Hence, $v^{4}(\{4,5\})=4$. Finally, it is easy to see that $\hat{s}_{\{4,5\}}\in\mathcal{AS}^{4}(\{1,2,4,5\})$ is also an optimal schedule for $\{1,2,4,5\}$. Then, $v^{4}(\{1,2,4,5\})=4$, and the game $(N,v^{4})$ is not superadditive (nor convex).
\end{example}

As the structure of the game does not help to study balancedness, in our main result we show that a particular allocation of the total cost savings lays in the core. Given a unit time open shop scheduling problem with initial schedule $(N,M,s_0)$ and in view of Algorithm \ref{Algo:AA} (see also Example \ref{ex:algo}), for all $j\in M$ there exists an optimal schedule, that we call $\hat{s}_{N}^{j}$, for $N$ such that its unique compatible scheme $\hat{\sigma}\in\Sigma$ satisfies $\hat{\sigma}^{j}=\sigma_{0}^{j}$ and, moreover, machine $j$ does not incur any idle time (operations on machine $j$ are processed continuously) according to $\hat{s}_{N}^{j}$. For any $j\in M$ we introduce the $j$-based allocation $\mu^{j}(N,M,s_0)\in\mathbb{R}^{N}$ by:
\begin{equation}
\label{def:all_rule_j}
\mu^{j}_{i}(N,M,s_0)=c_{i}(s_{0})-c_{i}(\hat{s}^{j}_{N})=C_{i}(s_{0})-C_{i}(\hat{s}^{j}_{N}) \quad \text{for all}\,\, i\in N.
\end{equation}
Given $j\in M$, this allocation assigns to each player the difference between her initial waiting cost and the cost associated with the optimal schedule $\hat{s}^{j}_{N}$ for $N$. It is easy to see that this allocation is efficient, since
$$\sum\limits_{i\in N}\mu^{j}_{i}(N,M,s_0)=\sum\limits_{i\in N}(c_{i}(s_{0})-c_{i}(\hat{s}^{j}_{N}))=c_{N}(s_{0})-c_{N}(\hat{s}^{j}_{N})=v^{4}(N),$$
where the last equality follows from the fact that $\hat{s}_{N}^{j}$ is optimal for $N$. However, as the next example shows, $\mu^{j}(N,M,s_0)$ does not need to satisfy $\mu^{j}_{i}(N,M,s_0)\geq v^{4}(i)$ for all $i\in N$, and hence need not be a core element.
\begin{example}
\label{ex:all_rule_j_notcore}
Consider $(N,M,s_0)$ with $N=\{1,2,3,4\}$, $M=\{1,2\}$, and the initial schedule $s_{0}$ as follows:
\begin{center}
\begin{tabular}{ c|c|c|c|c| } 
\cline{2-5}
$m_{1}$ & 1 & 2 & 3 & 4 \\  \cline{2-5} 
$m_{2}$ & 3 & 4 & 1 & 2 \\  \cline{2-5}
\end{tabular}.
\end{center}
Let $j=\{1\}$. Then, it is easy to check that $\hat{s}^{1}_{N}$ is
\begin{center}
\begin{tabular}{ c|c|c|c|c| } 
\cline{2-5}
$m_{1}$ & 1 & 2 & 3 & 4 \\  \cline{2-5} 
$m_{2}$ & 2 & 1 & 4 & 3 \\  \cline{2-5}
\end{tabular}.
\end{center}
Now, consider $i=\{3\}$. Then,
$$\mu^{1}_{3}(N,M,s_0)=C_{3}(s_{0})-C_{3}(\hat{s}^{1}_{N})=3-4=-1<v^{4}(\{3\})=0.$$
Hence, the allocation $\mu^{1}(N,M,s_0)$ is not a core element of the game $(N,v^{4})$.
\end{example}

Given $(N,M,s_{0})$, in Theorem \ref{teo} we will show that although the $j$-based allocation does not need to be a core element, surprisingly, the average of all $\mu^{j}(N,M,s_{0})$ always belongs to the core.
For any $(N,M,s_0)$, the average machine-based allocation rule $\overline{\mu}(N,M,s_0)$ is defined by
\begin{equation}
\label{def:all_rule_ave}
\bar{\mu}(N,M,s_0)=\frac{1}{m}\sum\limits_{j\in M}\mu^{j}(N,M,s_0).
\end{equation}

In order to show our main result, let us first prove a technical lemma.

\begin{lemma}
\label{lemma}
Let $(N,M,s_0)$ be a unit time open shop scheduling problem with initial schedule. Then, for all $\emptyset\neq T\subset N$, $i\in T$, it holds 
\begin{equation}
\label{CR}
\frac{1}{m}\sum\limits_{j\in M}\Big(C_{i}(s_0)-\left\lceil \frac{C^{j}_{i}(\hat{s}_{T})}{m} \right\rceil m\Big) \geq C_{i}(s_0)-C_{i}(\hat{s}_{T}),
\end{equation}
where $\hat{s}_{T}\in\mathcal{AS}^{4}(T)$ is optimal for $T$.
\end{lemma}

\begin{proof}
Let $(N,M,s_0)$, $T\subset N$, and $i\in T$. Let $j^{*}\in M$ be such that $C_{i}(\hat{s}_{T})=C^{j^{*}}_{i}(\hat{s}_{T})$. Then, clearly $C^{j}_{i}(\hat{s}_{T})\leq C^{j^{*}}_{i}(\hat{s}_{T})$ for all $j\in M$. Moreover, $$C^{j^{*}}_{i}(\hat{s}_{T})\leq \left \lceil \frac{C^{j*}_{i}(\hat{s}_{T})}{m}\right\rceil m.$$ Here, if according to $\hat{s}_{T}$ we make consecutive blocks of $m$ units of time from the moment at which the system starts processing, $\left\lceil\frac{C^{j^{*}}_{i}(\hat{s}_{T})}{m}\right\rceil$ stands for the block in which the operation of player $i$ is processed on machine $j^{*}$.

We distinguish two cases:
\vspace{0.5cm}
\\
\noindent\textit{Case 1:} $C^{j^{*}}_{i}(\hat{s}_{T})=\left \lceil \frac{C^{j^{*}}_{i}(\hat{s}_{T})}{m}\right\rceil m$.\vspace{0.2cm}
\\
Graphically:
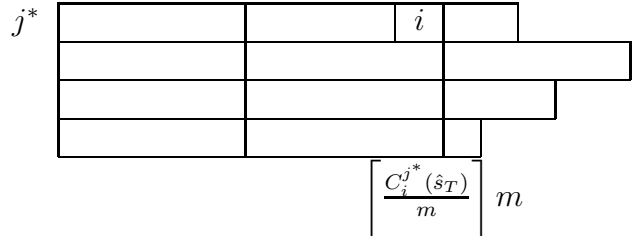
\begin{center}
\begin{tabular}{ c|ccccc|ccccc|ccccc }
 \cline{2-13} 
$j^{*}$ & & & & & & & & & & \multicolumn{1}{|c|}{$i$} &  & \multicolumn{1}{c|}{} &  & &     \\ \cline{2-16} 
  & & & & & & & & & & & & & & & \multicolumn{1}{c|}{}  \\ \cline{2-16} 
  & & & & & & & & & & &  &  & \multicolumn{1}{c|}{} & &   \\ \cline{2-14} 
  & & & & & & & & & &  &\multicolumn{1}{c|}{} & & & &     \\ \cline{2-12} 
\end{tabular}
\begin{tabular}{ cccccccccccccccc }
 & & & & & & & & & &   & \multicolumn{1}{c}{$\left\lceil \frac{C^{j^{*}}_{i}(\hat{s}_{T})}{m}\right\rceil m$} &  &  & &     \\ 
\end{tabular}
\end{center}
It is straightforward to see from the definition of $j^{*}\in M$ that 
\begin{equation}
\label{case 1: j*}
C_{i}(s_0)-\left\lceil \frac{C^{j^{*}}_{i}(\hat{s}_{T})}{m}\right\rceil m=C_{i}(s_0)-C_i(\hat{s}_{T}),
\end{equation}
and for any other machine $j\in M$
\begin{align}
\label{case1: j}
C_{i}(s_{0})-\left \lceil \frac{C^{j}_{i}(\hat{s}_{T})}{m}\right\rceil m &\geq C_{i}(s_{0})-\left \lceil \frac{C^{j^{*}}_{i}(\hat{s}_{T})}{m}\right\rceil m \\
&=C_{i}(s_{0})-C_{i}(\hat{s}_{T}). \nonumber
\end{align}
Then, from (\ref{case 1: j*}) and (\ref{case1: j})
\begin{align*}
\frac{1}{m}\sum\limits_{j\in M}\left( C_{i}(s_{0})-\left \lceil \frac{C^{j}_{i}(\hat{s}_{T})}{m}\right\rceil m \right) & \geq \frac{1}{m}\sum\limits_{j\in M} \left( C_{i}(s_{0})-C_{i} (\hat{s}_{T})\right)  \\
& =C_{i}(s_{0})-C_{i}(\hat{s}_{T}),
\end{align*}
which finishes Case 1.
\vspace{0.2cm}
\\
\noindent\textit{Case 2:} $\label{case2:statement} C^{j^{*}}_{i}(\hat{s}_{T})<\left \lceil \frac{C_{i}^{j^{*}}(\hat{s}_{T})}{m}\right\rceil  m$.
\vspace{0.2cm}\\
Graphically:
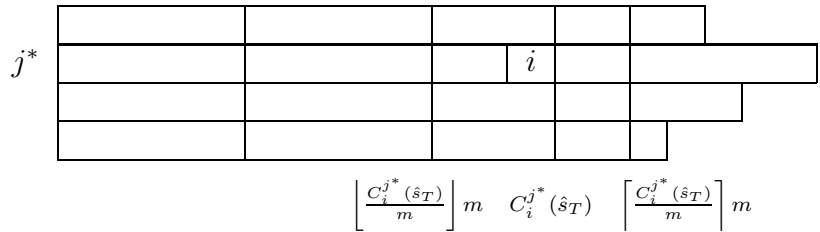
\begin{center}
\begin{tabular}{ c|ccccc|ccccc|ccccc|ccccc }
 \cline{2-18} 
 & & & & & & & & & & & & & \multicolumn{1}{c|}{} & &  &  & \multicolumn{1}{c|}{} &  & &     \\ \cline{2-21} 
$j^{*}$ & & & & &  & & & & & & & & \multicolumn{1}{|c|}{$i$} & & & & & & & \multicolumn{1}{c|}{}  \\ \cline{2-21} 
  & & & & & & & & & & & & & \multicolumn{1}{c|}{} & & &  &  & \multicolumn{1}{c|}{} & &   \\ \cline{2-19} 
  & & & & & & & & & & & & & \multicolumn{1}{c|}{} & &  &\multicolumn{1}{c|}{} & & & &     \\ \cline{2-17}
\end{tabular}
\end{center}
\hspace{6 cm} {\scriptsize{$\left\lfloor \frac{C^{j^{*}}_{i}(\hat{s}_{T})}{m}\right\rfloor m$}} \, {\scriptsize{$C^{j^{*}}_{i}(\hat{s}_{T})$}} \, {\scriptsize{$\left\lceil \frac{C^{j^{*}}_{i}(\hat{s}_{T})}{m}\right\rceil m$}}
\vspace{0.4cm}\\
By definition of $j^{*}$, 
$$\left\lceil\frac{C_{i}^{j}(\hat{s}_{T})}{m}\right\rceil \leq \left \lceil \frac{C_{i}^{j^{*}}(\hat{s}_{T})}{m}\right\rceil$$ for all $j\in M$.

Let $J^{*}=\Big\{j\in M : \left\lceil\frac{C_{i}^{j}(\hat{s}_{T})}{m}\right\rceil=\left\lceil\frac{C_{i}^{j^{*}}(\hat{s}_{T})}{m}\right\rceil \Big\}$. Note that $J^{*}\neq\emptyset$ since $j^{*}\in J^{*}$. On the other hand, if $j\in M\setminus J^{*}$, then $$\left\lceil\frac{C_{i}^{j}(\hat{s}_{T})}{m}\right\rceil < \left \lceil \frac{C_{i}^{j^{*}}(\hat{s}_{T})}{m}\right\rceil$$ or equivalently 
\begin{equation}
\label{case2:j not in J bar}
\left\lceil\frac{C_{i}^{j}(\hat{s}_{T})}{m}\right\rceil \leq \left \lceil \frac{C_{i}^{j^{*}}(\hat{s}_{T})}{m}\right\rceil -1
\end{equation}
To establish an upper bound for $| J^{*}|$, notice that all $j\in J^{*}$, $j\neq j^{*}$, process the operation of job $i$ in the same block as $j^{*}$. Additionally, we have $C^{j}_{i}(\hat{s}_{T})<C^{j^{*}}_{i}(\hat{s}_{T})$. So, there are as much as $C^{j^{*}}_{i}(\hat{s}_{T})-\left\lfloor \frac{C^{j^{*}}_{i}(\hat{s}_{T})}{m} \right \rfloor m$ different machines in $J^{*}$, i.e.
\begin{equation}
\label{case2:card J}
|J^{*}|\leq C^{j^{*}}_{i}(\hat{s}_{T})-\left\lfloor \frac{C^{j^{*}}_{i}(\hat{s}_{T})}{m}\right\rfloor m.
\end{equation}
Hence,
\begin{align}
\label{case2: card M-J}
|M\setminus J^{*}| & \geq m-\Big( C_{i}^{j^{*}}(\hat{s}_{T})- \left\lfloor \frac{C^{j^{*}}_{i}(\hat{s}_{T})}{m}\right \rfloor m\Big) \nonumber \\
& = \left\lceil \frac{C^{j^{*}}_{i}(\hat{s}_{T})}{m}\right \rceil m - C_{i}^{j^{*}}(\hat{s}_{T}),
\end{align}
since $\left\lfloor \frac{C^{j^{*}}_{i}(\hat{s}_{T})}{m}\right\rfloor = \left\lceil \frac{C^{j^{*}}_{i}(\hat{s}_{T})}{m}\right\rceil -1$. Then,

\begin{align*}
& \sum\limits_{j\in M} \Big( C_{i}(s_{0})-\left\lceil \frac{C^{j}_{i}(\hat{s}_{T})}{m}\right\rceil m\Big)= \\
=& \sum\limits_{j\in J^{*}} \Big( C_{i}(s_{0})- \left\lceil \frac{C^{j}_{i}(\hat{s}_{T})}{m}\right\rceil m \Big) + \sum\limits_{j\in M\setminus J^{*}}\Big( C_{i}(s_{0})-  \left\lceil \frac{C^{j}_{i}(\hat{s}_{T})}{m}\right\rceil m\Big) \\
\geq &\sum\limits_{j\in J^{*}} \Big( C_{i}(s_{0})- \left\lceil \frac{C^{j^{*}}_{i}(\hat{s}_{T})}{m}\right\rceil m \Big)+\sum\limits_{j\in M\setminus J^{*}} \Big( C_{i}(s_{0})-  \left\lceil \frac{C^{j^{*}}_{i}(\hat{s}_{T})}{m}\right\rceil m +m \Big)
\end{align*}
\begin{align*}
=&|J^{*}| \Big(C_{i}(s_{0})-\left\lceil \frac{C^{j^{*}}_{i}(\hat{s}_{T})}{m}\right\rceil m \Big)+\mid M\setminus J^{*}\mid \Big( C_{i}(s_{0})-\left\lceil \frac{C^{j^{*}}_{i}(\hat{s}_{T})}{m}\right\rceil m+m \Big) \\
\geq & \Big( C^{j^{*}}_{i} (\hat{s}_{T})- \left\lfloor \frac{C^{j^{*}}_{i}(\hat{s}_{T})}{m}\right\rfloor m \Big) \Big( C_{i}(s_{0})-\left\lceil \frac{C^{j^{*}}_{i}(\hat{s}_{T})}{m}\right\rceil m \Big)+\\
+ & \Big( \left\lceil \frac{C^{j^{*}}_{i}(\hat{s}_{T})}{m}  \right\rceil m-C^{j^{*}}_{i}(\hat{s}_{T}) \Big) \Big( C_{i}(s_{0})-\left\lceil \frac{C^{j^{*}}_{i}(\hat{s}_{T})}{m}  \right\rceil m+m\Big)\\
=&\Big( C_{i}(s_{0})-\left\lceil \frac{C^{j^{*}}_{i}(\hat{s}_{T})}{m}  \right\rceil m \Big)\Bigg( m\Big( \left\lceil \frac{C^{j^{*}}_{i}(\hat{s}_{T})}{m}  \right\rceil - \left\lfloor \frac{C^{j^{*}}_{i}(\hat{s}_{T})}{m}\right\rfloor  \Big)\Bigg)+ m \Big( \left\lceil \frac{C^{j^{*}}_{i}(\hat{s}_{T})}{m}  \right\rceil m - C^{j^{*}}_{i}(\hat{s}_{T})\Big)\\\
=& m \Bigg[ \Big( C_{i} (s_{0})-\left\lceil \frac{C^{j^{*}}_{i}(\hat{s}_{T})}{m}  \right\rceil m \Big)+\Big( \left\lceil \frac{C^{j^{*}}_{i}(\hat{s}_{T})}{m}  \right\rceil m -C_{i}^{j^{*}}(\hat{s}_{T})\Big)\Bigg] \\
=& m\Big( C_{i}(s_{0})-C^{j^{*}}_{i}(\hat{s}_{T})\Big)=m(C_{i}(s_{0})-C_{i}(\hat{s}_{T})),
\end{align*}
where the first inequality follows from the definition of $J^{*}$ and (\ref{case2:j not in J bar}). The second inequality follows from (\ref{case2:card J}), (\ref{case2: card M-J}), and the observation that for all $j\in J^{*}$ and $j'\in M\setminus J^{*}$, $C_{i}(s_{0})-\left\lceil \frac{C^{j^{*}}_{i}(\hat{s}_{T})}{m}  \right\rceil m<C_{i}(s_{0})-\left\lceil \frac{C^{j^{*}}_{i}(\hat{s}_{T})}{m} \right\rceil m+m$. The last equality is by definition of $j^{*}$.

Consequently, (\ref{CR}) holds and this finishes Case 2.
\end{proof}
Now, we can state the main result of the paper.
\begin{theorem}
\label{teo}
Let $(N,M,s_0)$ be a unit time open shop scheduling problem with initial schedule. Then, $\bar{\mu}(N,M,s_0)\in C(v^{4})$.
\end{theorem}

\begin{proof}
Let $(N,M,s_0)$, $\mu^{j}(N,M,s_0)=\mu^{j}$, and $\bar{\mu}(N,M,s_0)=\bar{\mu}$. First, we show that the allocation rule $\bar{\mu}\in\mathbb{R}^{N}$ is efficient.
\begin{align*}
\bar{\mu}(N)=\sum\limits_{i\in N}\bar{\mu}_{i}=&\sum\limits_{i\in N}\frac{1}{m}\sum\limits_{j\in M}\big(C_{i}(s_{0})-C_{i}(\hat{s}^{j}_{N})\big)\\
&=\frac{1}{m}\sum\limits_{j\in M}\sum\limits_{i\in N}\left(C_{i}(s_{0})-C_{i}(\hat{s}^{j}_{N})\right)\\
&=\frac{1}{m}\sum\limits_{j\in M}c_{N}(s_{0})-c_{N}(\hat{s}^{j}_{N})\\
&=\frac{1}{m}m\left[c_{N}(s_{0})-c_{N}(\hat{s}_{N})\right]=v^{4}(N),
\end{align*}
where the fifth equality follows from $c_{N}(\hat{s}^{j}_{N})=c_{N}(\hat{s}^{j'}_{N})=c_{N}(\hat{s}_{N})$ for all $j,j'\in M$, $j\neq j'$. It remains to prove $\bar{\mu}(T)\geq v(T)$ for all $T\subset N$.

\begin{align*}
\bar{\mu}(T)&=\sum\limits_{i\in T}\bar{\mu}_{i}=\sum\limits_{i\in T}\frac{1}{m}\sum\limits_{j\in M}\mu^{j}_{i}\\
&=\sum\limits_{i\in T}\frac{1}{m}\sum\limits_{j\in M} \Big( C_{i}(s_{0})-C_{i}(\hat{s}^{j}_{N})\Big)\\
&=\sum\limits_{i\in T}\frac{1}{m}\sum\limits_{j\in M} \Big( C_{i}(s_{0})-\left\lceil \frac{\sigma_{0}^{j}(i)}{m}\right\rceil m \Big) \\
&=\frac{1}{m}\sum\limits_{j\in M}\Bigg( \sum\limits_{i\in T} C_{i}(s_{0})-\sum\limits_{i\in T}\left\lceil \frac{\sigma_{0}^{j}(i)}{m} \right\rceil m \Bigg)\\
&=\frac{1}{m}\sum\limits_{j\in M}\Bigg(\sum\limits_{i\in T} C_{i}(s_{0})-\sum\limits_{R\in T/ \sigma^{j}_{0}}\sum\limits_{i\in R} \left\lceil \frac{\sigma_{0}^{j}(i)}{m}\right\rceil m \Bigg)\\
&\geq \frac{1}{m}\sum\limits_{j\in M}\Bigg(\sum\limits_{i\in T} C_{i}(s_{0})-\sum\limits_{R\in T/ \sigma^{j}_{0}}\sum\limits_{i\in R} \left\lceil \frac{C_{i}^{j}(\hat{s}_{T})}{m} \right\rceil m \Bigg)
\end{align*}
\begin{align*}
&= \frac{1}{m}\sum\limits_{j\in M}\Bigg(\sum\limits_{i\in T} C_{i}(s_{0})-\sum\limits_{i\in T}\left\lceil \frac{C_{i}^{j}(\hat{s}_{T})}{m}\right\rceil m \Bigg)\\
&=\sum\limits_{i\in T}\frac{1}{m}\sum\limits_{j\in M}\Bigg( C_{i}(s_{0})-\left\lceil \frac{C_{i}^{j}(\hat{s}_{T})}{m}\right\rceil m \Bigg)\\
&\geq \sum\limits_{i\in T} C_{i}(s_0)-C_{i}(\hat{s}_{T})=v(T).
\end{align*}
The third equality follows from the fact that the unique compatible scheme $\hat{\sigma}$ with $\hat{s}_{N}^{j}$ satisfies $\hat{\sigma}^{j}=\sigma_{0}^{j}$ and moreover, in $\hat{s}_{N}^{j}$ operations are processed continuously on machine $j$. So, in view of Algorithm \ref{Algo:AA} (see also Example \ref{ex:algo}), $C_{i}(\hat{s}^{j}_{N})=\left\lceil \frac{{\sigma}^{j}_{0}(i)}{m}\right\rceil m$. The first inequality holds due to for every $j\in M$, and by definition of $\mathcal{AS}^{4}(T)$, players in $R\in T/\sigma_{0}^{j}$ can only switch their positions with other players in $R$. Then, if $R=\{i_1,i_2,\ldots ,i_r\}$ and $\hat{\sigma}_{T}$ is the unique optimal scheme compatible with $\hat{s}_{T}$, it holds $\{\sigma_{0}^{j}(i_{1}),\ldots ,\sigma_{0}^{j}(i_{r})\}=\{\hat\sigma^{j}_{T}(i_{1}),\ldots ,\hat{\sigma}^{j}_{T}(i_{r})\}.$ Moreover, for $i_k\in R$, $C^{j}_{i_{k}}(\hat{s}_{T})\geq \hat{\sigma}_{T}^{j}(i_{k})$, and hence $\sum\limits_{i\in R}\left\lceil \frac{\sigma^{j}_{0}(i)}{m}\right\rceil m\leq \sum\limits_{i\in R}\left\lceil \frac{C^{j}_{i}(\hat{s}_{T})}{m}\right\rceil m$. The last inequality follows from Lemma \ref{lemma}.
\end{proof}
\section{Relaxed unit open shop scheduling games}
\label{sec:relax}

In this last section, we study whether the balancedness result still holds when we relax the definition of admissible rearrangements for a coalition.  We follow the same approach introduced in \cite{cetal93} and later used by \cite{s06}. In particular, we would like to allow the players of a coalition to jump over players outside the coalition if such a switch does not hurt them. That is to say, it does not imply an increase in their completion times. 

As \cite{cetal93}, we first change condition (\ref{relax1}) by the following weaker condition on the schemes:

Let $(M,N,s_0)$, $\emptyset\neq T\subseteq N$, and $\sigma_{0}\in\Sigma$ denotes the unique scheme compatible with $s_0$. We say that a schedule $s$ is admissible for $T$ if for all $i\in N\setminus T$ and $j\in M$ it holds

\begin{equation}
\label{relax1'}
\sigma^{j}_{0}(i)=\sigma^{j}(i), \tag{$1^{\prime}$}
\end{equation}
where $\sigma$ is the unique scheme compatible with $s$.

Of course, to prevent hurting players in $N\setminus T$ we will combine (\ref{relax1'}) with (\ref{relax2}), (\ref{relax3}), and (\ref{relax4}) to obtain $\mathcal{AS}^{2^{\prime}}(T)$, $\mathcal{AS}^{3^{\prime}}(T)$, and $\mathcal{AS}^{4^{\prime}}(T)$, respectively. For each new set of admissible rearrangements, we can associate the corresponding cooperative game $(N,v^{2^{\prime}})$, $(N,v^{3^{\prime}})$, and $(N,v^{4^{\prime}})$. Proposition \ref{prop:games_prime} provides the relation among these games.

\begin{proposition}
\label{prop:games_prime}
Let $(N,M,s_0)$ be a unit open shop scheduling problem with initial schedule. Then, it holds 
\begin{align*}
v^{2^{\prime}}(N)=v^{3^{\prime}}(N)=v^{4^{\prime}}(N) &\\
v^{2^{\prime}}(T)\leq v^{3^{\prime}}(T)\leq v^{4^{\prime}}(T) & \qquad  \forall T\subset N.
\end{align*}
\end{proposition}

In the next theorem, we observe that the game $(N,v^{4^{\prime}})$ is balanced, and consequently $(N,v^{3^{\prime}})$ and $(N,v^{2^{\prime}})$, too.

\begin{theorem}
\label{teo_prime}
Let $(N,M,s_0)$ be a unit time open shop scheduling problem with initial schedule. Then, $\bar{\mu}(N,M,s_0)\in C(v^{4^{\prime}})$.
\end{theorem}

\begin{proof}
We first observe that Lemma \ref{lemma} still holds if we consider $\mathcal{AS}^{4^{\prime}}(T)$ instead of $\mathcal{AS}^{4}(T)$. This follows because we do not use any argument based on the definition of admissible rearrangements. We only employ the specifications of open shop scheduling problems. In particular, the restriction that two operations of the same job cannot be processed simultaneously on two different machines. 

Moreover, Theorem \ref{teo} also holds. It is enough to follow the lines of the proof together with the observation that for $T\subset N$ and $\hat{s}_{T}\in\mathcal{AS}^{4^{\prime}}(T)$ we have 
\begin{equation}
\label{relax_teo}
\sum\limits_{i\in T}\left\lceil \frac{\sigma^{j}_{0}(i)}{m} \right\rceil m \leq \sum\limits_{i\in T}\left\lceil \frac{C^{j}_{i}(\hat{s}_{T})}{m} \right\rceil m.
\end{equation}
So, the inequality we proved for the maximally connected compenents of $T$ is also satisfied for the full coalition $T$. In fact, if $T=\{i_1,\ldots, i_t\}$ and $\hat{\sigma}_{T}$ is the unique optimal scheme compatible with $\hat{s}_{T}$, by condition (\ref{relax1'}), we have $\{\sigma_{0}^{j}(i_{1}),\ldots, \sigma_{0}^{j}(i_{t})\}=\{\hat{\sigma}_{T}^{j}(i_{1}),\ldots, \hat{\sigma}_{T}^{j}(i_{t})\}$. Moreover, for every $i_{k}\in T$ $C^{j}_{i_{k}}(\hat{s}_{T})\geq \hat{\sigma}_{T}^{j}(i_{k})$ and consequently (\ref{relax_teo}) holds.
\end{proof}

To finish, we provide a counterexample to illustrate that further relaxations on the admissible rearrangements leads to games that violate balancedness. Let $T\subset N$, we now admit any schedule for $T$ except if it hurts players in $N\setminus T$, without imposing any condition on the associated scheme. Let $\overline{\mathcal{AS}}^{k}(T)$ for $k=\{2,3,4\}$ be the set of admissible rearrangements for a coalition $T\subseteq N$ that satisfy only condition $(k)$, but not necessarily (\ref{relax1}) or (\ref{relax1'}). For $k=\{2,3,4\}$, by $(N,\overline{v}^{k})$ we denote the game associated with a unit open shop scheduling problem with initial schedule $(N,M,s_0)$ where the set of admissible rearrangements is $\overline{\mathcal{AS}}^{k}(T)$ for any $T\subset N$. The relation among such games is stated in the next proposition.
\begin{proposition}
\label{prop:games}
Let $(N,M,s_0)$ be a unit open shop scheduling problem with initial schedule. Then, it holds 
\begin{align*}
\overline{v}^{2}(N)=\overline{v}^{3}(N)=\overline{v}^{4}(N) &\\
\overline{v}^{2}(T)\leq \overline{v}^{3}(T)\leq \overline{v}^{4}(T) & \qquad  \forall T\subset N.
\end{align*}
\end{proposition}
Next, we show the non-balancedness result.
\begin{proposition}
Let $(N,M,s_0)$ be a unit time open shop scheduling problem. Then, the associated game $(N,\overline{v}^{k})$ may not be balanced for any $k=\{2,3,4\}$.
\end{proposition}
\begin{proof}
In view of Proposition \ref{prop:games}, it is enough to show that there is $(N,M,s_0)$ such that $(N,\overline{v}^{2})$ is not balanced. Let $N=\{1,2,3\}$, $M=\{1,2\}$, and the initial schedule $s_0$ as follows:
\begin{center}
\begin{tabular}{ c|c|c|c|c|c| } 
\cline{2-6} 
$m_{1}$ & 1 & 2 & 3 & & \\  \cline{2-6} 
$m_{2}$ & & 1 & & 3 & 2 \\  \cline{2-6}
\end{tabular}.
\end{center}
Let $T=\{2\}$. It is easy to check that $\hat{s}_{\{2\}}\in\overline{\mathcal{AS}}^{2}(\{2\})$ is
\begin{center}
\begin{tabular}{ c|c|c|c|c|c| } 
\cline{2-5} 
$m_{1}$ & 1 & 2 & 3 & \\  \cline{2-5} 
$m_{2}$ & 2 & 1 & & 3 \\  \cline{2-5}
\end{tabular}.
\end{center}
Hence, $\overline{v}^{2}(\{2\})=3$. Let $T=\{3\}$, it is easy to check that $\hat{s}_{\{3\}}\in\overline{\mathcal{AS}}^{2}(\{3\})$ is
\begin{center}
\begin{tabular}{ c|c|c|c|c|c|c| } 
\cline{2-6} 
$m_{1}$ & 1 & 2 & 3 & & \\  \cline{2-6} 
$m_{2}$ & 3 & 1 & & & 2 \\  \cline{2-6}
\end{tabular},
\end{center}
and hence, $\overline{v}^{2}(\{3\})=1$. Finally, $\hat{s}_{N}$ is
\begin{center}
\begin{tabular}{ c|c|c|c|c|c| } 
\cline{2-5} 
$m_{1}$ & 1 & 2 & 3 & \\  \cline{2-5} 
$m_{2}$ & 2 & 1 & & 3 \\  \cline{2-5}
\end{tabular}.
\end{center}
Hence, $\overline{v}^{2}(N)=3$, and there does not exist an allocation $x\in\mathbb{R}^{3}$ that satisfies $x_{1}+x_{2}+x_{3}=3$, $x_{2}\ge 3$, $x_{3}\geq 1$, and $x_{1}\geq \overline{v}^{2}(\{1\})$ since, obviously $\overline{v}^{2}(\{1\})\geq 0$. Therefore, $C(\overline{v}^{2})=\emptyset$.
\end{proof}
\bibliographystyle{te}
\bibliography{openshop_july_30}

\end{document}